\documentclass[a4paper,10pt]{scrartcl}
\usepackage[utf8]{inputenc}
\usepackage[english]{babel}
\usepackage[fixlanguage]{babelbib}
\usepackage{cite}
\usepackage{amssymb, amsmath, amstext, amsopn, amsthm, amscd, amsxtra, amsfonts, commath}
\usepackage{yfonts, mathrsfs}
\usepackage{colonequals}
\usepackage{enumerate}
\usepackage{graphicx}%
\usepackage{colonequals}
\usepackage[all]{xy}
\usepackage{todonotes}
\usepackage{tikz-cd}
\usepackage{hyperref}
\hypersetup{
urlcolor = red,
colorlinks = true,
linkcolor = blue,
citecolor = blue,
linktocpage = true,
pdftitle = {Lie Theory for Asymptotic Symmetries in General Relativity: The BMS Group},
pdfauthor = {David Prinz, Alexander Schmeding},
bookmarksopen = true,
bookmarksopenlevel = 1,
unicode = true,
hypertexnames =false
}%
\usepackage{cleveref}
\usepackage{tabularx}
\usepackage{booktabs}
\swapnumbers
\newtheoremstyle{standard}
{16pt} 
{16pt} 
{} 
{} 
{\bfseries}
{} 
{ } 
{{\thmname{#1~}}{\thmnumber{#2.}}\thmnote{~(#3)}} 
\newtheoremstyle{kursiv}
{16pt} 
{16pt} 
{\itshape} 
{} 
{\bfseries}
{} 
{ } 
{{\thmname{#1~}}{\thmnumber{#2.}}\thmnote{~(#3)}} 
\theoremstyle{standard}
\newtheorem{defn}{Definition}[section]

\newtheorem{rem}[defn]{Remark}

\newtheorem{setup}[defn]{}

\theoremstyle{kursiv}

\newtheorem{prop}[defn]{Proposition}
\newtheorem{cor}[defn]{Corollary}

\newtheorem{conj}[defn]{Conjecture}

\newcommand{\Lf}{\ensuremath{\mathbf{L}}}

\newcommand{\bC}{\ensuremath{\mathbb{C}}}
\newcommand{\eC}{\ensuremath{\hat{\mathbb{C}}}}

\newcommand{\R}{\ensuremath{\mathbb{R}}}

\newcommand{\N}{\ensuremath{\mathbb{N}}}

\newcommand{\K}{\ensuremath{\mathbb{K}}}
\newcommand{\SSS}{\ensuremath{\mathbb{S}}}

\DeclareMathOperator{\one}{\mathbf{1}}

\DeclareMathOperator{\Diff}{Diff}
\DeclareMathOperator{\Evol}{Evol}

\DeclareMathOperator{\id}{id}

\DeclareMathOperator{\PSL}{\mathrm{PSL}}
\DeclareMathOperator{\Lor}{\mathrm{SO}^+ (3,1)}
\DeclareMathOperator{\SL}{\mathrm{SL}}

\DeclareMathOperator{\Aut}{Aut}
\newcommand{\Frechet}{Fr\'{e}chet}

\newcommand{\LB}[1][\cdot \hspace{1pt} , \cdot]{\left[\hspace{1pt} #1 \hspace{1pt} \right]}
\DeclareMathOperator{\evol}{\mathrm{evol}}
\DeclareMathOperator{\BMS}{BMS}
\DeclareMathOperator{\gBMS}{gBMS}

\newcommand{\scri}{\mathscr{I}}

\newcommand{\conext}{\widehat{M}}
\newcommand{\conmet}{\hat{g}}
\newcommand{\z}{\mathbf{z}}

\begin{document}
\title{Lie Theory for Asymptotic Symmetries in General Relativity: The BMS Group}
\author{David Prinz\footnote{Humboldt-Universität zu Berlin, Germany: \href{mailto:prinz@math.hu-berlin.de}{prinz@math.hu-berlin.de}}\hspace{2mm} and Alexander Schmeding\footnote{Nord universitetet i Levanger, Norway: \href{mailto:alexander.schmeding@nord.no}{alexander.schmeding@nord.no}
}%
}
\date{January 26, 2022}
{\let\newpage\relax\maketitle}

\begin{abstract}
	We study the Lie group structure of asymptotic symmetry groups in General Relativity from the viewpoint of infinite-dimensional geometry. To this end, we review the geometric definition of asymptotic simplicity and emptiness due to Penrose and the coordinate-wise definition of asymptotic flatness due to Bondi et al. Then we construct the Lie group structure of the Bondi--Metzner--Sachs (BMS) group and discuss its Lie theoretic properties. We find that the BMS group is regular in the sense of Milnor, but not real analytic. This motivates us to conjecture that it is not locally exponential. Finally, we verify the Trotter property as well as the commutator property. As an outlook, we comment on the situation of related asymptotic symmetry groups. In particular, the much more involved situation of the Newman--Unti group is highlighted, which will be studied in future work.
\end{abstract}



\medskip


\tableofcontents

\section{Introduction and statement of results} 

The experimental verification of gravitational waves at LIGO in 2016 \cite{Abbott_et-al} renewed interest in their theoretical predictions. A natural setup to study gravitational waves are asymptotically flat spacetimes. This is due to the fact that, as described below in detail, asymptotically flat spacetimes correspond to isolated gravitational systems. Thus, they are ideally suited to study outgoing gravitational waves of such systems. A natural question concerning asymptotic flatness is how the asymptotic symmetry group is related to the Poincar\'{e} group \cite{Bondi_VdBurg_Metzner,Sachs}. Additionally to that, the asymptotic symmetry groups provide useful insight into the gravitational \(S\)-matrix via soft-scattering theorems for the corresponding Feynman rules \cite{Weinberg,Strominger,He_et-al}, cf.\ \cite{Prinz} and the references therein. In the present article and \cite{Prinz_Schmeding_2}, we provide the Lie group theoretic foundations for these asymptotic symmetry groups.\smallskip

\textbf{Acknowledgments:} The authors wish to thank G.\ Barnich and H.\ Friedrich for pointing us towards literature we were unaware of and helpful discussions on the subject of this work. DP is supported by the `Kolleg Mathematik Physik Berlin'.

\subsection{Physical motivation}

Asymptotically flat spacetimes model the geometry of isolated gravitational systems, such as solar systems or even complete galaxies. More precisely, the Einstein field equations relate energy-matter contributions to the spacetime geometry as follows:
\begin{equation} \label{eqn:efe}
R_{\mu \nu} - \frac{1}{2} R g_{\mu \nu} + \Lambda g_{\mu \nu} = \kappa T_{\mu \nu} \, ,
\end{equation}
where \(R_{\mu \nu}\) and \(R\) are the Ricci tensor and scalar, respectively, \(\Lambda\) and \(\kappa := 8 \pi G\) are the cosmological and the Einstein gravitational constant, respectively, and \(T_{\mu \nu}\) is the energy-momentum tensor. Thus, if the matter distribution is compactly supported, with the radiation of energy via electromagnetic waves as the only exception, then the spacetime geometry becomes asymptotically Ricci-flat. Furthermore, as the Weyl-curvature describes the radiation of energy via gravitational waves, the spacetime-geometry becomes asymptotically flat, as gravitational radiation is governed by the inverse square law on large scales. Thus, asymptotically flat spacetimes approach Minkowski spacetime `at infinity', and it seems reasonable that the corresponding asymptotic symmetry groups are related to the Poincar\'{e} group. It turns out, however, that the corresponding asymptotic symmetry groups are much richer in their structure: Similar to the Poincar\'{e} group, they are given as the semidirect product of generalized translations, called `supertranslations', and, depending on the boundary conditions, the Lorentz group or a generalization thereof, similarly called `superrotations'.\footnote{We remark that these findings, and thus the corresponding nomenclature, was ahead of supersymmetry and super geometry. Thus, while the name suggests a link between these two subjects, they are in fact unrelated.} Contrary to the translation group of the Poincar\'{e} group, the supertranslations are an infinite dimensional Lie group. We explain this surprising finding after introducing the necessary context. To this end, we start --- opposed to the historic development --- with the geometric approach of Penrose via a conformal extension of spacetime, leading to the notions of asymptotically simple and asymptotically empty spacetimes \cite{Penrose_1,Penrose_2,Penrose_3,Penrose_4}. Then we proceed with the pioneering works of Bondi, van der Burg and Metzner \cite{Bondi_VdBurg_Metzner} and Sachs \cite{Sachs} via their coordinate-wise definition, which can be seen as the choice of a special diffeomorphism `gauge-fixing' additionally to the geometric restrictions described by the construction due to Penrose.

\subsection{Mathematical motivation}

The relationship between symmetries and conservation laws has a long standing history, with Noether's Theorem (1918) as an important milestone \cite{Noether}. Symmetries can be modeled via the action of Lie groups. Especially in the context of variations and gauge theories (e.g.\ via BRST cohomology) these symmetries are typically studied infinitesimally, i.e.\ via the action of the corresponding Lie algebras. In finite dimensions there is a strong connection between Lie groups and their Lie algebras. For example if \(G\) is a matrix Lie group with Lie algebra \(\mathfrak{g}\). Then, the Lie exponential is the matrix exponential
\begin{align}
\operatorname{exp} \colon \quad \mathfrak{g} \to G \, , \quad Z \mapsto \sum_{k = 0}^\infty \frac{1}{k!} Z^k \, , \label{eqn:exp}
\end{align}
and it induces a diffeomorphism (exponential coordinates) between a $0$-neighborhood of the Lie algebra and an identity-neighborhood in $G$. The infinite-dimensional case is much more subtle as this correspondence fails in general. The Lie exponential is not guaranteed to be locally surjective \cite{Milnor,Grabowski}. This happens for example for the diffeomorphism group \(\operatorname{Diff} (M)\) of a smooth manifold \(M\) (see \cite{Sch15} and \cite{KaM97}). To recover Lie theoretic statements in the infinite-dimensional setting, an evolution operator for smooth Lie algebra valued curves becomes important. Sticking to common notation in the physics literature, this evolution operator is 
\begin{align}
\operatorname{Evol} \colon C^\infty (I, \mathfrak{g}) \to C^\infty (I, G) \, , \quad \gamma (t) \mapsto \mathcal{T} \operatorname{exp} \left ( \int_0^t \gamma (\tau) \dif \tau \right ) \, , \label{eqn:evol}
\end{align}
where \(I \subset \mathbb{R}\) is a compact interval, and the right hand side a time ordered exponential. In mathematical terms, the evolution operator computes the solution to an initial value problem given by the Lie type equation
\begin{equation}
\begin{cases} \dot{\eta} &= \eta . \gamma,\\ \eta (0) &=1\end{cases} \label{eqn:Lie_type_ODE}
\end{equation}
(where the lower dot signifies the derivative of the left translation by $\eta$). A Lie group which admits a smooth evolution operator in the sense of Equation~\eqref{eqn:evol} is called regular in the sense of Milnor \cite{Milnor}, cf.\ \cite[Chapter 38]{KaM97} and \cite{Neeb06}. We remark that the evolution operator from Equation~(\ref{eqn:evol}) specializes to the Lie exponential from Equation~\eqref{eqn:exp} for constant curves \(\gamma (t) \equiv Z\). Note that the Lie groups we are about to consider are modeled on locally convex spaces which are not Banach spaces. Thus there is no general solution theory for ordinary differential equations, such as Equation~\eqref{eqn:Lie_type_ODE}, on these spaces guaranteeing regularity. 

Mathematically speaking the structure of the BMS group is quite simple. Let us explain this in the three dimensional toy case, i.e.\ for the group $\BMS_3$. It is well known that $\BMS_3$ can be identified as the semidirect product $\mathcal{V} (\SSS^1) \rtimes_{\text{Ad}}  \Diff (\SSS^1)$ where the diffeomorphism group $\Diff (\SSS^1)$ acts on the vector fields via the adjoint action. Here the vector fields are interpreted as the Lie algebra of $\Diff (\SSS^1)$ and for the group structure as an abelian Lie group. Hence the Lie theory for $\BMS_3$ boils down to the Lie theory of $\Diff (\SSS^1)$ (and its adjoint action). In particular, $\BMS_3$ inherits several pathologies from $\Diff (\SSS^1)$: It is well known that $\Diff (\SSS^1)$ does not admit exponential coordinates (see e.g.\ \cite[Example 43.2]{KaM97}). Hence also the semidirect product $\BMS_3$ does not admit a (locally) surjective Lie group exponential. 
Note that for four dimensional spacetimes, the BMS group arises as a semidirect product of a finite dimensional Lie group with a complete metrizable space (also known as a \Frechet~space). Thus both parts of the semidirect product admit good Lie group exponentials and we discuss the problem of the existence of exponential coordinates for the BMS group in \Cref{sect:BMS}. Let us remark that for $\BMS_3$ and the BMS group it is essential to work with smooth supertranslations (i.e. the abelian part of the semidirect product is a space of smooth mappings). This settles the question as to how smooth the supertranslations need to be (see e.g.\ \cite{McCar72b,McCar72a,Girardello:1974sq} where finitely often differentiable mappings and of $L^2$-mappings are discussed in the context of BMS representation theory). If the BMS should be a Lie group endowing the supertranslations with a natural topology,\footnote{The adjective ``natural'' for the topology and the Lie group structure arise here because one could trivially endow every set with the discrete topology. This turns the BMS group into a $0$-dimensional Lie group. Clearly this Lie group structure is neither natural nor interesting.} then the supertranslations must necessarily be smooth mappings. Indeed, in view of a prospective complexification of the BMS as a Lie group (see below in the statement of results), it appears to be desirable from a mathematical point of view to require supertranslations to be even real analytic. Note that similar remarks apply to the supertranslations for BMS like groups for other spacetime dimensions, \cite{Melas:2017jzb}. Finally, we discuss the validity of familiar identities, such as the Baker--Campbell--Hausdorff formula and Trotter product formula for the BMS group.

\subsection{Original results of the present article}

For the present article we investigate the BMS group and some of its generalizations as infinite-dimensional Lie groups. We establish that the $\BMS$ group is an infinite-dimensional Lie group modeled on a \Frechet~space. 
This result is not surprising, as it has been long known to physicists \cite{Sachs,McCar72a} that the BMS group is a semidirect product
\begin{equation}
\BMS = \mathcal{S} \rtimes \Lor
\end{equation}
of a function space $\mathcal{S} =C^\infty (\SSS^2)$ (viewed as an abelian Lie group) and the Lorentz group \(\Lor\), which is in particular a finite-dimensional Lie group. The novelty here is that we establish smoothness of the group action on the infinite-dimensional space of supertranslations $\mathcal{S}$. Since the supertranslations can not be modeled on a Banach space, one needs to replace the usual calculus as it makes no sense in our more general setting. A natural notion of smooth map is obtained here by requiring the existence and continuity of (all) iterated directional derivative. This yields a well behaved calculus of differentiable mappings, called Bastiani calculus (see \Cref{app:idim} for more information). 

To our knowledge an investigation of the Lie theoretic properties of the BMS group has not been attempted in the literature. 
Among our results is that the BMS group is regular (again this is not surprising in view of the semidirect product structure). As a consequence, we obtain the validity of the Trotter product formula and the commutator formula on the BMS group. 
These formulae are important tools in the representation theory of Lie groups. In the finite-dimensional setting they can be deduced from the Baker--Campbell--Hausdorff (BCH) formula. This leads us to more surprising results of our investigation: The BMS group is \textbf{not} an analytic Lie group. In a nutshell, the reason for this is that the group product incorporates function evaluations of smooth (but not necessarily analytic) mappings. We refer to \Cref{prop:BMS:notanalytic} below for more details. In particular, this entails that
\begin{itemize}
 \item the well known Baker--Campbell--Hausdorff series does not provide a local model for the Lie group multiplication, and
 \item there can not be a complexification of the BMS group which continues the real BMS group multiplication as a complex (infinite-dimensional) Lie group. (See also \cite{McCar92} on complexifications of the BMS group.)
\end{itemize}
Moreover, we remark that due to this defect either the Lie group exponential does not provide a local diffeomorphism onto an identity neighborhood or the BCH-series can not converge on any neighborhood of $0$ in the Lie algebra. 

Finally, we discuss the Lie group structure of an extension of the BMS group, called the generalized BMS group (or gBMS). It can be identified with the semidirect product $\mathcal{S} \rtimes \Diff (\SSS^2)$ and behaves quite similarly (when it comes to the Lie theory) to the BMS group. 

\subsection{An outlook to the Newman--Unti group} \label{ssec:outlook_NU}

Several extensions and generalizations of the BMS group have been proposed (see e.g.\ \cite{Ruzz20} for an overview). Beyond the generalized BMS group there is the extended BMS group proposed by Barnich and collaborators, \cite{Barnich:2009se} (we will not discuss this extension in the present paper). Another possibility to enlarge the BMS group is to replace the supertranslations in the semidirect product by a certain subset $\mathcal{N}$ of $C^\infty (\R \times \SSS^2)$. The Newman--Unti group is then the semidirect product (with respect to a suitable action)
\begin{equation}
\operatorname{NU} = \mathcal{N} \rtimes \Lor \, .
\end{equation}
Note that contrary to the BMS case, the function space $C^\infty (\R \times \SSS^2)$ consists now of functions on a non-compact manifold. This looks like a minor complication but indeed it changes the whole theory: For non-compact manifolds \(M\), the algebra of smooth functions $C^\infty (M)$ admits at least three qualitatively different choices of function space topologies (the compact open $C^\infty$-topology, the Whitney topologies and the fine very strong topology). Only two of these topologies lead to either a topological vector space (compact open topology) or to a manifold structure (fine very strong topology) on the function space. This complication does not appear for the supertranslations $\mathcal{S} = C^\infty (\SSS^2)$ as the compactness of the sphere $\SSS^2$ implies that all these topologies coincide. To avoid lengthy discussions of function space topologies only to discuss the Newman--Unti group, we refrain from this analysis in the current paper. 

Instead, the topology and Lie group structure of the Newman--Unti group will be discussed in future work \cite{Prinz_Schmeding_2}. There we will show that with respect to 
\begin{itemize}
 \item the compact open $C^\infty$ topology the NU group becomes a topological group (but apparently not a Lie group as it lacks a manifold structure), 
 \item the fine very strong topology the unit component of the NU group becomes an infinite-dimensional Lie group, while the full NU group does not become a Lie group.
\end{itemize}

\section{Asymptotically flat spacetimes}

We start this article with an introduction to different definitions of asymptotic flatness. As already mentioned in the introduction, we start opposed to the historical development with the construction due to Penrose and discuss the pioneering works of Bondi et al.\ and Sachs afterwards. The reason therefore is that the conformal extension due to Penrose is more fundamental, as it directly describes the spacetime-geometry, whereas the construction of Bondi et al.\ and Sachs is more restrictive, by using specific coordinate representations. We refer to \cite{Ashtekar,Friedrich:2017cjg} for excellent overview articles.

\subsection{Penrose's conformal extension} \label{ssec:Penrose_conformal_extention}

A very geometrical approach to this subject is given via Penrose's conformal extension \cite{Penrose_1,Penrose_2,Penrose_3,Penrose_4}. Basically the idea is that, given a spacetime \((M,g)\), we call it asymptotically simple if there exists a conformal embedding into a so-called `extended spacetime' \((\widehat{M},\hat{g})\), which is a manifold with boundary that represents the points `at infinity'. To this end, we start with the definition of a spacetime:

\begin{defn}[Spacetime] \label{def:spacetime}
	Let \((M,g)\) be a Lorentzian manifold. We call \((M,g)\) a spacetime, if it is smooth, connected, 4-dimensional and time-orientable.\footnote{We fix the spacetime-dimension solely for a cleaner presentation. Our main results are actually valid in arbitrary spacetime-dimensions, i.e.\ for any compact manifold $K$ and conformal group $\text{Conf} (K)$.}
\end{defn}

With the following construction, we want to unify spacetimes that approach Minkowski spacetime `at infinity'. For this, we construct the boundary
\begin{equation}
\scri := i_+ \sqcup \scri_+ \sqcup i_0 \sqcup \scri_- \sqcup i_- \, ,
\end{equation}
where \(i_+\) and \(i_-\) correspond to timelike infinity, \(\scri_+\) and \(\scri_-\) to lightlike infinity and finally \(i_0\) to spacelike infinity. This is then set as the boundary of a so-called `extended spacetime' (or sometimes also `unphysical spacetime') \((\widehat{M},\hat{g})\), and then related to the spacetime \((M,g)\) via a conformal embedding \(\iota \colon M \hookrightarrow \widehat{M}\), as follows:

\begin{defn}[Asymptotically simple (and empty) spacetime] \label{defn:Penrose_ase_spacetime}
	Let \((M,g)\) be an oriented and causal spacetime. We call \((M,g)\) an asymptotically simple spacetime, if it admits a conformal extension \(\big ( \conext,\conmet \big )\) in the sense of Penrose: That is, if there exists an embedding \(\iota \colon M \hookrightarrow \conext\) and a smooth function \(\varsigma \in C^\infty \big ( \conext \big )\), such that:
	\begin{enumerate}
		\item \(\conext\) is a manifold with interior \(\iota \left ( M \right )\) and boundary \(\scri\), i.e.\ \(\conext \cong \iota \left ( M \right ) \sqcup \scri\)
		\item \(\eval[2]{\varsigma}_{\iota \left ( M \right )} > 0\), \(\eval[2]{\varsigma}_{\scri} \equiv 0\) and \(\eval[2]{\dif \varsigma}_{\scri} \not \equiv 0\); additionally \(\iota_* g \equiv \varsigma^2 \conmet\)
		\item Each null geodesic of \(\big ( \conext,\conmet \big )\) has two distinct endpoints on \(\scri\)
	\end{enumerate}
	We call \((M,g)\) an asymptotically simple and empty spacetime, if additionally:\footnote{This condition can be modified to allow electromagnetic radiation near \(\scri\). We remark that asymptotically simple and empty spacetimes are also called asymptotically flat.}
	\begin{enumerate}
		\setcounter{enumi}{3}
		\item \(\eval[2]{\left ( R_{\mu \nu} \right )}_{\iota^{-1} ( \widehat{O} )} \equiv 0\), where \(\widehat{O} \subset \conext\) is an open neighborhood of \(\scri \subset \conext\)
	\end{enumerate}
\end{defn}

Using this construction, the BMS and Newman--Unti groups can be seen as diffeomorphisms acting on \(\scri\), cf.\ \cite{Schmidt_Walker_Sommers}. Furthermore, the geometry of asymptotically simple and empty spacetimes can be characterized as follows:

\begin{prop} \label{prop:ase-parallelizable}
	Let \((M,g)\) be an asymptotically simple and empty spacetime. Then \((M,g)\) is globally hyperbolic and thus parallelizable.
\end{prop}

\begin{proof}
	The first part of the statement, i.e.\ that \((M,g)\) is globally hyperbolic, is a classical result due to Ellis and Hawking \cite[Proposition 6.9.2]{Hawking_Ellis}. We conclude the second part, i.e.\ that \((M,g)\) is parallelizable, by noting that we have additionally assumed spacetimes to be four-dimensional: Thus, being globally hyperbolic, the space-submanifold is three-dimensional and thus parallelizable, as it is orientable by assumption.
\end{proof}

Moreover, in the case of a vanishing cosmological constant \(\Lambda = 0\), the topology of \(\scri\) can be characterized as follows:

\begin{prop} \label{prop:topology_scri}
	Let \((M,g)\) be an asymptotically simple spacetime with vanishing cosmological constant \(\Lambda = 0\). Then the two components representing lightlike infinity \(\scri_\pm \supset \scri\) are both homeomorphic to \(\mathbb{R} \times \mathbb{S}^2\).
\end{prop}

\begin{proof}
	This is a classical result due to Newman \cite[Corollary 2]{Newman}.
\end{proof}

Finally, we also mention the viewpoint from symplectic geometry to the above concepts via the Hamiltonian formalism of General Relativity \cite{Ashtekar_Streubel}.

\subsection{Bondi et al.'s coordinate-wise definition} \label{ssec:BMS_coortinate_definition}

The original approach to this subject was due to the pioneering works of Bondi, van der Burg and Metzner \cite{Bondi_VdBurg_Metzner} and Sachs \cite{Sachs} via the following coordinate-wise definition. As it turned out later using Penrose's conformal extension, asymptotically flat spacetimes are parallelizable (\propref{prop:ase-parallelizable}). Thus, the following coordinate functions can be defined globally, modulo possible singularities. Furthermore, we remark that they are constructed for spherically symmetric situations, which motivates the use of spherical coordinates for the spatial coordinates:

\begin{defn}[BMS coordinate functions] \label{defn:bms_coordinates}
	Let \((M,g)\) be an asymptotically simple spacetime with globally defined coordinate functions \(x^\alpha \colon M \to \mathbb{R}^4\), denoted via \(x^\alpha := (t,x,y,z)\). Then we introduce the so-called BMS coordinate functions \(y^\alpha \colon M \to \mathbb{R} \times [0, \infty) \times \mathbb{S}^2\), denoted via \(y^\alpha = (u,r,\vartheta,\varphi)\), as follows: We transform the spatial coordinates \((x,y,z)\) to spherical coordinates \((r,\vartheta,\varphi)\) via the relations
	\begin{subequations}
	\begin{align}
	r & := \sqrt{x^2 + y^2 + z^2} \, , \\
	\vartheta & := \arccos \left ( \frac{z}{r} \right )\\
	\text{and} \quad \varphi & := \arctan \left ( \frac{y}{x} \right ) \, ,
	\intertext{and then combine the timelike coordinate \(t\) together with the radial coordinate \(r\) to the lightlike coordinate}
	u & := t - r \, .
	\end{align}
	\end{subequations}
	Given these coordinates and combining the angles \(z^a := (\vartheta,\varphi)\), the metric reads as follows:
	\begin{subequations}
	\begin{align}
	g_{\mu \nu} \dif x^\mu \otimes \dif x^\nu & \equiv - \frac{V}{r} e^{2 \beta} \dif u \otimes \dif u - e^{2 \beta} \left ( \dif u \otimes \dif r + \dif r \otimes \dif u \right ) \\ & \phantom{\equiv} + r^2 h_{a b} ( \dif z^a - U^a \dif u ) \otimes ( \dif z^b - U^b \dif u ) \, ,
	\intertext{where \(h_{a b}\) is the metric on the (deformed) unit sphere, i.e.}
	\begin{split}
	h_{a b} \dif z^a \otimes \dif z^b & \equiv \cosh (2 \delta) \left ( e^{2 \gamma} \dif \vartheta \otimes \dif \vartheta + e^{- 2 \gamma} \sin^2 (\vartheta) \dif \varphi \otimes \dif \varphi \right ) \\ & \phantom{\equiv} + \sin (\vartheta) \sinh(2 \delta) \left ( \dif \vartheta \otimes \dif \varphi + \dif \varphi \otimes \dif \vartheta \right ) \, .
	\end{split}
	\end{align}
	\end{subequations}
	Here, we have expressed the metric degrees of freedom via real functions on the spacetime \(V, \beta, \gamma, \delta \in C^\infty ( M, \mathbb{R} )\) and a vector field on the unit sphere \(U \in \mathfrak{X} ( S^2 )\).
\end{defn}

\begin{rem} \label{rem:diffeomorphism_gauge_fixing}
	BMS coordinate functions are an optimal parametrization of the metric, i.e.\ every degree of a general Lorentzian globally hyperbolic metric is parametrized via a function, cf.\ \cite{Maedler_Winicour}. In that, they can be seen as a diffeomorphism gauge fixing.
\end{rem}

\begin{defn}[BMS asymptotic flatness] \label{defn:bms_asymptotic_flatness}
	Given the coordinate functions from \defnref{defn:bms_coordinates}, the spacetime \((M,g)\) is called asymptotically flat, if for all \((u, r, \vartheta, \varphi) \ni y^\alpha\) fixed, we have the following radial asymptotic behavior:
	\begin{subequations} \label{eqns:fall-off_properties}
	\begin{align}
	 \lim_{r \to \infty} V / r & = 1 + \mathcal{O} ( 1/r ) \, , \\
	 \lim_{r \to \infty} h_{a b} & = q_{a b} + \mathcal{O} ( 1/r ) \quad \text{with} \quad q_{a b} = \operatorname{diag} \left ( 1, \sin^2 (\vartheta) \right ) \, , \\
	 \lim_{r \to \infty} \beta & = \mathcal{O} ( 1/r^3 )\\
	\text{and} \quad \lim_{r \to \infty} U^a & = \mathcal{O} ( 1/r^3 ) \, .
	\end{align}
	\end{subequations}
\end{defn}

\begin{rem}
	Comparing \defnref{defn:Penrose_ase_spacetime} to \defnref{defn:bms_asymptotic_flatness}, we remark that the fall-off properties from Equations~(\ref{eqns:fall-off_properties}) translate into a smoothness-condition at \(\scri\). Additionally, the BMS framework provides a diffeomorphism gauge fixing, as was discussed in \remref{rem:diffeomorphism_gauge_fixing}. Thus, the corresponding asymptotic symmetry group, which will be called BMS group, needs to respect this coordinate choice.
\end{rem}

\section{Asymptotic symmetry groups}

Asymptotic symmetry groups are subgroups of the diffeomorphism group that preserve the chosen boundary condition and gauge fixing. We focus in this article on the BMS group and their generalizations that preserve the conditions given in Subsection~\ref{ssec:BMS_coortinate_definition}. Our aim is to establish (infinite-dimensional) Lie group structures on the asymptotic symmetry groups. For readers who are not familiar with calculus beyond Banach spaces, we have compiled the basic definitions in \Cref{app:idim} (and we suggest to review them before continuing).

\subsection{General constructions}

Before we begin, let us recall two general constructions which will be used throughout the following sections:

\begin{setup}\label{prelim:ck-top}
 We will encounter spaces of differentiable mappings as infinite-dimensional manifolds. Let us repeat some important definitions and properties of these manifolds.
 For $M,N$ (smooth) manifolds, we denote by $C^k(M,N)$ the space of $k$-times continuously differentiable mappings from $M$ to $N$, where $k \in \N_0 \cup \{\infty\}$. 
 If nothing else is said, we topologize $C^k (M,N)$ with the compact open $C^k$-topology. This is the topology turning the mapping
 $$C^k (M,N) \rightarrow \prod_{\ell \in \N_0 , \ell \leq k} C (T^\ell M, T^\ell N) ,\quad f \mapsto (T^\ell f)_{0 \leq \ell \leq k}$$
 into a topological embedding. Here the spaces on the right hand side are topologized with the compact open topology and $T^\ell$ denotes the $\ell$-fold iterated tangent functor.
 Recall from \cite[Appendix A]{AaGaS20} that this topology turns $C^k (M,N)$ into a Banach (resp.~\Frechet) manifold for $k \in \N_0$ (resp.~$k=\infty$) if $M$ is compact and $N$ is a finite dimensional manifold. 
 Moreover, one can prove that if $N$ is a locally convex vector space, then also $C^k(M,N)$ is a locally convex vector space with the pointwise operations.
 
 Now we finally recall from \cite[Lemma A.10]{AaGaS20} that for $M$ compact and $N$ a finite dimensional manifold the manifold $C^k(M,N)$ is canonical, i.e.~a mapping 
 $$h \colon A \rightarrow C^k (M,N) \text{ for any smooth manifold } A$$
 is of class $C^\ell$ if and only if the adjoint map $h^\wedge \colon A \times M \rightarrow N, (a,m)\mapsto h(a)(m)$ is a $C^{\ell,k}$-map.
 This means that $h^\wedge$ is $\ell$-times continuously differentiable with respect to the $A$-component of the product and each of these differentials is then $k$-times differentiable with respect to the $M$-component (cf.\ \Cref{defn:CRS}).
\end{setup}

Finally, we recall the concept of a semidirect product (all asymptotic symmetry groups discussed in this article will turn out to be semidirect products).
Let us first recall from \cite[Lemma 2.2.3]{HaN12} the notion of a semidirect product 

\begin{setup}[Semidirect product of groups]
 Let $N$ and $H$ be groups and $\Aut(N)$ the group of automorphisms of $N$. Assume that $\delta \colon H \rightarrow \Aut (N)$ is a group homomorphism. Then we define a multiplication on $N \times H$ by
 \begin{align}\label{semidirect}
  (n,h) (m,g) := (n\delta(h)(m), hg)
 \end{align}
 This multiplication turns $N \times H$ into a group denoted by $N \rtimes_\delta H$, where $N \cong N \times \{e\}$ is a normal subgroup and $H \cong \{e\}\times H$ is a subgroup. Furthermore, each element $x \in N \rtimes_\delta H$ has a unique representation as $x = nh, n \in N, h \in H$.
 
 If $H,N$ are Lie groups (or analytic Lie groups) and $\delta^\wedge \colon H \times N \rightarrow H, (h,n) \mapsto \delta (h)(n)$ is smooth (analytic)\footnote{If $N$ is finite-dimensional, it suffices to require that $\delta$ is a Lie group morphism. In general, there is no Lie group structure on $\Aut (N)$ which guarantees smoothness of the group operation in $N \rtimes_\delta H$.}, then $N \rtimes_\delta H$ is a Lie group (analytic Lie group). Its Lie algebra is a semi-direct product of Lie algebras.
\end{setup}

\subsection{The BMS group}\label{sect:BMS}

The BMS group can be identified with conformal automorphisms of the future infinity boundary $\scri^+$ of the extended spacetime $\hat{M}$. Identifying the boundary as $\scri^+ = \SSS^2 \times \R$, we can think of $\SSS^2$ as the Riemann sphere with standard spherical coordinates $(\theta,\phi) \in [0,\pi[ \times [0,2\pi[$.
Recall from \cite[Appendix B.2 and B.3]{Oli02} that we can use the stereographic projections to construct charts for the sphere $\SSS^2$ as a complex $1$-dimensional manifold:
$$\text{st}\colon \SSS^2 \setminus \{(0,0,1)\} \rightarrow \bC \cong \R^2,\quad \text{st}(\theta,\phi) := \zeta := e^{i\phi} \cot \left( \frac{\theta}{2}\right).$$
Extending the stereographic projection to a diffeomorphism $\kappa \colon \SSS^2 \rightarrow \bC \cup \{\infty\}$ via $\kappa (\z) = \begin{cases}
 \text{st}(\z)& \z \in \SSS^2 \setminus \{(1,0,0)\}\\
\infty & \text{else}
\end{cases}$
of $\SSS^2$ with the extended complex plane $\eC := \bC \cup \{\infty\}$ (aka the Riemann sphere). Traditionally, the resulting complex coordinates are denoted by $\mathbf{z} = (\zeta ,\overline{\zeta})$ (where the second is the complex conjugate of the first). However, since in this section we will work globally and (mostly) suppress explicit coordinate expressions, we will continue to write $\z$ for the variable from $\SSS^2$. Let us now collect some well known facts about the complex structure of the sphere and its interplay with the conformal group:

\begin{setup}\label{smooth:complex}
 Using the isomorphism with the extended complex plane, the orientation preserving conformal automorphisms of $\SSS^2$ (with the standard conformal structure) can be identified as the group of M\"{o}bius transformations of determinant $1$ of the Riemann sphere:
\begin{equation}\label{Mobius}
\mathcal{M} := \left\{\bC \cup \{\infty\}\ni \zeta \mapsto \frac{a\zeta +b}{c\zeta +d} \middle| a,b,c,d \in \bC , ad-bc=1\right\}.
\end{equation}
where (as usual) for $c\neq 0$ the definition of $f \in \mathcal{M}$ is extended to the Riemann sphere by $f(\infty) = a/c$ and $f(-d/c)=\infty$). Note that two elements in $\mathcal{M}$ are equal if their actions on the Riemann sphere coincide. Thus $\mathcal{M} \cong \text{SL} (2,\bC) / \{\pm I\} \equalscolon \PSL (2,\bC)$ as Lie groups. In particular, $\text{SL} (2,\bC)$ is the universal cover of $\mathcal{M}$ (which is locally diffeomorphic to $\mathcal{M}$). Hence we can conveniently construct smooth or analytic maps to $\mathcal{M}$ by composing such maps with values in $\text{SL}(2,\bC)$ with the canonical quotient map. See \cite[Appendix B]{Oli02} for a thorough discussion of M\"{o}bius transformations and the resulting M\"{o}bius group. 

It is essential to recall that $\mathcal{M}$ is a complex Lie group (cf.~\cite[p.113]{Kna02} and \cite[B.4]{Oli02}) which acts by biholomorphic transformations on the Riemann sphere. In particular, the resulting action is a complex Lie group action 
$$\gamma_\bC\colon \mathcal{M} \times \eC \rightarrow \eC , (\Lambda,\zeta) \mapsto \frac{a\zeta+b}{c\zeta + d} , \quad\text{for } \Lambda = \begin{bmatrix} a & b \\ c & d \end{bmatrix},
$$
i.e.~it is holomorphic. One can prove \cite[2.3]{Schott08} that every conformal orientation preserving isomorphism of $\SSS^2$ corresponds (up to composition with the stereographic projections) to a M\"{o}bius transformation. Under this identification we obtain a holomorphic mapping which computes the conformal factor (cf.~\cite[p.220]{Oli02}) at every point of the Riemann sphere:  
$$K_\bC \colon \mathcal{M} \times \eC \rightarrow ]0,\infty[ ,\quad K(\Lambda , \zeta) := 
\frac{1+\lVert \zeta\rVert^2}{\lVert a\zeta + b\rVert^2+\lVert c\zeta + d\rVert^2} \text{ for } \Lambda = \begin{bmatrix} a & b \\ c & d\end{bmatrix}.$$
 Observe that $K_{\bC}$ only takes values in $]0,\infty[$ and we can naturally extend the definition of the conformal factor to the point $\zeta = \infty$ by setting $K_\bC (\Lambda,\infty) = 1 /(|a|^2+|c|^2)$.
\end{setup}

Usually, we will wish to consider the BMS group as a real Lie group. Thus we have to forget the complex structure on the conformal group. Indeed we would like to interpret a M\"{o}bius transformation as an orientation preserving conformal automorphism of $\SSS^2$.

\begin{setup}
Forgetting the complex structure on $\mathcal{M}$ we can treat it as a real analytic Lie group, cf.~e.g.~\cite[I.13]{Kna02}. This yields the natural real Lie group structure for the conformal group of the sphere, which is isomorphic (as a Lie group) to the proper orthochronous Lorentz group $\Lor$.\footnote{We suppress the group isomorphism and think of elements in $\Lor$ as conformal transformations on $\SSS^2$.}
To stress that we are working with conformal mappings we shall define for every $f \in \Lor$ the associated M\"{o}bius transformation $\kappa \circ f \circ \kappa^{-1}$ and denote by $\Lambda_f \in \mathcal{M}$ its matrix representation. 
Note that the composition with the diffeomorphism $\kappa$ induces the isomorphism between $\Lor$ and $\mathcal{M}$ which we are implicitly exploiting here.
\end{setup}

\begin{setup}\label{setup:realana}
Now $\Lor \cong \mathcal{M}$ is a real analytic Lie group. The action of $\text{SL}(2,\bC)$ is analytic, whence we obtain a real analytic (and in particular smooth) action on $\SSS^2$ via
$$\gamma \colon \Lor \times \SSS^2 \rightarrow \SSS^2, \quad (f, \z ) \mapsto f(\z).$$
Moreover, the map associating the conformal factor to an element of $\text{SL}(2,\bC)$ is smooth (even real analytic) and descents to a smooth (real analytic) map 
$$K \colon \Lor \times \SSS^2 \rightarrow ]0,\infty[ , (f,\z) \mapsto K_f (\z) := K_\bC (\Lambda_f, \kappa(\zeta))$$
Exploiting now the canonical manifold structure on the $C^k$-functions, \Cref{prelim:ck-top} we obtain a smooth map  
\begin{align*}
 \Psi^k \colon \Lor &\rightarrow C^k (\SSS^2,\SSS^2),\quad f \mapsto \left(\z \mapsto f(\z)\right), \forall k \in \N_0 \cup\{\infty\}, \\
 K^\vee \colon \Lor & \rightarrow C^\infty (\SSS^2,]0,\infty[),\quad f \mapsto (\z \mapsto K(f,\z)).
\end{align*}
Note that due to $\gamma$ being a group action, the map $\Psi := \Psi^\infty$ actually takes its image in $\Diff (\SSS^2)$ (which is an open subset of $C^\infty (\SSS^2,\SSS^2)$). 

Similarly we can invert the conformal factor and obtain a smooth map 
\begin{align*}
 \delta \colon \Lor  \rightarrow& C^\infty (\SSS^2 , \R),\quad f \mapsto K_f(\cdot)^{-1}.
\end{align*}
\end{setup}

\begin{setup}[A family of group actions]
 For each $k \in \N \cup \{\infty\}$ we define a right group action (cf.\ \cite[Proposition 6.4]{Alessio:2017lps}) via
 \begin{align}\label{top:gpact}
\sigma_k \colon C^k (\SSS^2,\R) \times \Lor \rightarrow C^k (\SSS^2,\R),\quad  \sigma_k (\alpha,f) := K_f(\cdot)^{-1} \cdot \alpha \circ f.  
 \end{align}
Here the spaces $C^k(\SSS^2,\R)$ are endowed with the compact open $C^k$-topology, \Cref{prelim:ck-top}.
To see that $\sigma_\infty$ is smooth, we rewrite it as follows:
$$\sigma_\infty (\alpha,f) = m_* (\delta (f) , \text{Comp} (\alpha , \Psi(f))),$$
where $\text{Comp} \colon C^\infty (\SSS^2,\SSS^2) \times C^\infty (\SSS^2 ,\R) \rightarrow C^\infty (\SSS^2 ,\R) , (f,\alpha) \mapsto \alpha \circ f$ is the composition of mappings. Composition is a smooth map by \cite[Theorem 11.4]{Michor80}. Now $m_*\colon C^\infty (\SSS^2 , \R)^2 \cong C^\infty (\SSS^2,\R \times \R) \rightarrow C^\infty (\SSS^2 , \R) , (\varphi ,\psi) \mapsto \varphi \cdot \psi$ is the pushforward with the multiplication in $\R$. Due to the compactness of $\SSS^2$ also the pushforward is smooth \cite[Corollary 10.14]{Michor80} and we conclude that $\sigma_\infty$ is smooth.

Note that for $k \in \N_0$ the action $\sigma_k$ is only continuous. Indeed as \cite[Theorem 11.4]{Michor80} shows, the derivative of the composition requires derivatives of the function composed from the left, whence composition can not be differentiable on spaces of finite differentiability.
\end{setup}

As these observations show, only the space of smooth functions supports a Lie group action \eqref{top:gpact} and we name this space according to the physics conventions.

\begin{defn}
 We define the \emph{space of (smooth) supertranslations} 
 $$\mathcal{S}:=C^\infty(\SSS^2):=C^\infty(\SSS^2,\R).$$
 Endowing it with the compact-open $C^\infty$-topology, we obtain a \Frechet -space which makes $\mathcal{S}$ into an abelian \Frechet~Lie group, called the \emph{group of supertranslations}. 
\end{defn}

It is well known that the BMS group can be described as a semidirect product of the orthochronous Lorentz group and the supertranslations, \cite{Sachs,McCar72a}. We follow the description in \cite[Proposition 6.4]{Alessio:2017lps} where the semidirect product structure is explicitly described:

\begin{setup}[The semidirect product of conformal maps with supertranslations]\label{setup:semidirect}
To ease notation we denote the action \eqref{top:gpact} of the conformal group of the unit sphere on the supertranslations $\mathcal{S}$ by
\begin{align}
\sigma \colon \mathcal{S} \times  \Lor \rightarrow  \mathcal{S} , \quad (\alpha,f) \mapsto (\z \mapsto \alpha(\z)\Lambda_f := K_f^{-1}(\z)\alpha(f(\z)) .
\end{align}
Observe that the action induces a group homomorphism
$$\sigma^\vee \colon \Lor \rightarrow \text{Aut}(C^\infty (\SSS^2)),\quad \phi \mapsto \sigma (\cdot, \phi )$$
and we can form the semidirect product $\mathcal{S} \rtimes_\sigma \Lor$. 

This group is a \Frechet~Lie group whose Lie algebra is $\mathcal{S} \rtimes_{\hat{\sigma}} \mathfrak{so} (1,3)$ where 
$$\hat{\sigma} \colon \mathfrak{so} (1,3) := \Lf (\Lor) \rightarrow \text{der} (C^\infty (\SSS^2))$$ is the derived representation to the representation of $\Lor $ on $\mathcal{S}$ given by $f . \alpha = T\sigma (f,\cdot) (\alpha)$ (cf.~\cite[Proposition 9.2.25]{HaN12}).

It has been shown in \cite{Alessio:2017lps} that the BMS group is the semidirect product of the supertranslations with the conformal group with respect to the action $\sigma$, in other words
 \begin{align} \label{eqn:BMS_semidirect_product}
  \BMS \cong \mathcal{S} \rtimes_{\sigma} \Lor
 \end{align}
 as a group. We deduce that the BMS group is an infinite-dimensional \Frechet~Lie group. Applying the definition of the multiplication for a semidirect product we obtain the following formulae for the product and the inverse in the BMS:
 \begin{align*}
  (F,\phi) \cdot (G,\psi) &= (F+ \sigma(G,\phi), \phi \circ \psi) =  (F+ K^{-1}_\phi \cdot G \circ \phi , \phi \circ \psi)\\
  (F,\psi)^{-1} &= (-\sigma (F,\psi^{-1}), \psi^{-1}) = (-K_\psi \cdot F\circ \psi^{-1}, \psi^{-1})
 \end{align*}

\end{setup}

\subsection{Properties of the BMS group}

We proceed by discussing the properties of the BMS group as an infinite-dimensional Lie group, as constructed in the previous subsection. We begin with a proof that the BMS-group is not a real analytic Lie group, which motivates our conjecture that the BMS group does not admit exponential coordinates.
For this it suffices to prove that the Lie group action
$$\sigma \colon \mathcal{S} \times \Lor \rightarrow \mathcal{S}$$
is not real analytic as it entails that the group multiplication can not be real analytic.
 It will turn out that $\sigma$ can not be analytic with respect to the variable from $\Lor$. To see this we will construct a suitable analytic curve into $\Lor$ which does not get mapped to an analytic curve under the action.

\begin{prop}\label{prop:BMS:notanalytic}
 The group action $\sigma \colon \mathcal{S} \times \Lor \rightarrow \mathcal{S}$ is not real analytic. Thus the BMS group is not real analytic.
\end{prop}

\begin{proof}
 For the proof we will identify $\SSS^2$ with the Riemann sphere $\eC$ via the (analytic) diffeomorphism $\kappa$. 
 
 \textbf{Step 1:} \emph{A smooth map on $\eC$ which is nowhere analytic near $0$}. 
 Pick a smooth but nowhere analytic function $\theta \colon \R \rightarrow \R$. The concrete function is irrelevant here, we just remark that many such functions exist, see \cite{Dar73}. Further, we pick a smooth cutoff function $\chi \colon \eC \rightarrow [0,\infty[$ such that $\chi(\zeta) \equiv 0$ if $|\zeta|> 43$ and $\chi (\zeta) \equiv 1$ if $|\zeta| \leq 42$. With this we define the smooth function
 $$h \colon \eC \rightarrow \R,\quad \zeta \mapsto \frac{\chi (\zeta)}{1+\text{Re} (\zeta)^2}\theta (\chi (\zeta) \text{Re} (\zeta)),$$
 where $\text{Re} (\zeta)$ is the real part of the complex number $\zeta \in \bC$ (note that $\zeta =\infty$ is irrelevant as we are cutting off with $\chi$).\\[1em]
 
 \textbf{Step 2:} \emph{An analytic curve in $\mathcal{S} \times \Lor$}.
 We consider the analytic one-parameter group $\R \rightarrow \text{SL}(2,\bC), t \mapsto \exp (tA)= \begin{bmatrix} 1 & t \\ 0 & 1 \end{bmatrix}$, where $A = \begin{bmatrix} 0 & 1 \\ 0 & 0 \end{bmatrix}$ and $\exp$ is the matrix exponential. Composing with the canonical quotient map, we obtain a real analytic curve $\eta_2 \colon \R \rightarrow \mathcal{M} , t \mapsto \begin{bmatrix} 1 &  t \\ 0 & 1 \end{bmatrix}$.
 Now $\eta_2 (t) (\zeta) = \zeta + t$ is just the translation of $\zeta$. Exploiting the real analytic diffeomorphism $\kappa \colon \SSS^2 \rightarrow \eC$, we obtain a real analytic curve 
 \begin{align*}
  \eta \colon \R \rightarrow \mathcal{S} \times \Lor, t \mapsto (h \circ \kappa, \kappa^{-1} \circ \eta_2(t) \circ \kappa).
 \end{align*}

 \textbf{Step 3:} \emph{The action is not analytic.}
 Consider first the mapping 
 $$\sigma^\wedge \colon \mathcal{S} \times \Lor \times \SSS^2 \rightarrow \R, (\alpha , f , \z) \mapsto K_f^{-1} (\z) \alpha (f(\z)).$$
 We plug in the curve from Step 2:
 \begin{align*}
  \sigma^\wedge (\eta (t),\kappa^{-1}(0)) &= K_{\eta_2(t)}^{-1}(\kappa^{-1}(0)) h(\kappa (\kappa^{-1} (\eta_2 (t)(\kappa (\kappa^{-1}(0)))) = K_{\eta_2 (t)}^{-1}(0) h(\eta_2 (t)(0)) \\
  &= (1+t^2) h(0+t) = (1+t^2) \frac{\chi(t)}{1+t^2} \theta (\chi(t) t) = \chi (t)\theta (\chi(t)t)
 \end{align*}
We deduce that whenever $|t|<42$ we have $\sigma^\wedge (\eta (t),\kappa^{-1}(0)) = \theta (t)$, so the composition of $\sigma^\wedge$ with our chosen analytic curve is not real analytic (at least on an open interval around $0$. To complete the proof, assume that $\sigma \circ \eta$ was analytic with values in $C^\infty (\SSS^2,\R)$. Then \cite[Lemma 42.12 (2)]{KaM97} shows that for every $\z \in \SSS^2$ the curve $(\sigma \circ \eta )^\wedge (\cdot ,\z) = \sigma^\wedge \circ \eta (\cdot)(\z) \colon \R \rightarrow \R$ admits a holomorphic extension, i.e.~is real analytic.\footnote{We remark that the cited sources uses real analyticity in the sense of convenient calculus. In general the concept of real analyticity differs from our preferred concept of real analyticity. However, since all manifolds and spaces involved are modeled on \Frechet~spaces, the two concepts are known to coincide.} However, we have already seen that for $\z =\kappa^{-1}(0))$ is not real analytic in an open neighborhood of $0$, whence $\sigma \circ \eta$ can not be real analytic. We deduce that the action $\sigma$ can not be real analytic.

The statement about the BMS group now follows directly from the definition of the multiplication for semidirect products.
\end{proof} 

\begin{rem}
 As the proof shows, the action can not be analytic since it contains an evaluation of a smooth function on the sphere. In the evaluated component this process can not be analytic since we can have smooth but nowhere analytic mappings. The problem could be remedied by passing to the subspace of real analytic functions inside of the supertranslations $\mathcal{S}$. Note however, that the natural topology on the analytic functions is not the subspace topology induced by the inclusion in $\mathcal{S}$. Indeed the topology of the analytic functions is not a \Frechet~topology and the construction of a Lie group structure for this subgroup of the BMS group would require refined tools for manifolds of analytic mappings (see e.g.~\cite{DaS15} for more information on such tools and the topology on manifolds of real analytic mappings).   
\end{rem}

We will now establish several Lie theoretic properties of the BMS group. Let us first recall the notion of regularity for infinite-dimensional Lie groups: A Lie group $G$ is $C^r$-semiregular, $r\in \N_0 \cup \{\infty\}$ if for every $C^r$-curve $\gamma \colon [0,1] \rightarrow \Lf (G))$ the initial value problem 
  \begin{equation}\begin{cases}
     \eta'(t) =  T_{\one} \rho_{\eta (t)} (\gamma(t)), \qquad \rho_{g}(h):= h g\\
     \eta (0) = \one
    \end{cases}
 \end{equation}
 has a unique $C^{r+1}$-solution $\Evol (\gamma) := \eta \colon [0,1] \rightarrow G$. Moreover, if in addition the evolution map 
 $\evol \colon C^r ([0,1],\Lf (G))  \rightarrow G, \gamma \mapsto \Evol (\gamma)(1)$ is smooth, we say $G$ is $C^r$-regular. If $G$ is $C^\infty$-regular (the weakest of the regularity conditions), $G$ is called \emph{regular (in the sense of Milnor)}. To employ advanced techniques in infinite-dimensional Lie theory, one needs to require regularity of the Lie groups involved, cf.\ \cite{hg2015}. Note that for a constant curve $\eta(t) \equiv = v\in \Lf(G)$ we simply recover the Lie group exponential $\evol (\eta (t)) = \exp (v)$. Thus every regular infinite-dimensional Lie group admits a Lie group exponential.

\begin{prop}\label{prop:BMS_regular}
 The BMS group is $C^0$-regular and in particular it is regular in the sense of Milnor.
\end{prop}

\begin{proof}
 We have seen in \Cref{setup:semidirect} that the BMS group is the semidirect product $\mathcal{S} \rtimes_\sigma \Lor$. Thus the BMS group is a Lie group extension of the conformal group by the supertranslations and with the canonical injection $\mathcal{S} \rightarrow \mathcal{S} \times \{\one\}$ and the projection we obtain a short exact sequence of Lie groups:
\begin{displaymath}\label{Lgp:ext}
 \begin{tikzcd}
0 \arrow[rr] & &  \mathcal{S} \arrow[rr, hookrightarrow] & &\BMS \arrow[rr, twoheadrightarrow] && \Lor \ar[rr]&& \mathbf{1}.
\end{tikzcd}   
 \end{displaymath}
 Now $C^k$-regularity of Lie groups is an extension property \cite{NaS13}, whence it suffices to prove that $\Lor$ and $\mathcal{S}$ are $C^0$-regular. However, $\Lor$ is $C^0$-regular as a finite dimensional Lie group. By definition $\mathcal{S}$ is a the abelian Lie group given by the \Frechet~space structure of $C^\infty (\SSS^2)$. It is well known \cite{Neeb06} that locally convex spaces are regular Lie groups if and only if they are Mackey complete. Now a \Frechet~space is complete and thus in particular Mackey complete. We deduce that $\mathcal{S}$ and thus also $\BMS$ are $C^0$-regular. 
\end{proof}

In general, the Lie group exponential need not induce a local diffeomorphism from the Lie algebra to an open identity neighborhood in an infinite-dimensional group (the most famous example of this pathology is the diffeomorphism group $\Diff (M)$ of a compact manifold, which is a Lie group whose Lie group exponential is not onto any neighborhood of the identity). Since we are now entering finer aspects of infinite-dimensional Lie theory, it is prudent to recall first several definitions:

\begin{defn}
 Let $G$ be a regular Lie group and $\Lf (G) = \mathfrak{g}$ be its Lie algebra. We say that 
 \begin{enumerate}
  \item $G$ is a \emph{(locally) exponential Lie groups} if the Lie group exponential $\exp\colon \mathfrak{g} \rightarrow G$ induces a (local) diffeomorphism onto an identity neighborhood.
  \item $\mathfrak{g}$ is a \emph{Baker--Campbell--Hausdorff (BCH) Lie algebra} if there exists an open $0$-neighborhood $U \subseteq \mathfrak{g}$ such that for $x, y \in U$ the BCH-series 
  $$\sum_{n=1}^\infty H_n(x,y) \text{ with } H_1 (x,y) = x+y, H_2 = \frac{1}{2} \LB[x,y] , \ldots ,$$
  converges and defines an analytic function $U \times U \rightarrow \mathfrak{g}, (x,y) \mapsto x \star y$ (note that if $\mathfrak{g}$ is a \Frechet~space, the analyticity of the product $x\star y$ is automatic (cf.~\cite[Definition IV.1.5]{Neeb06})
  \item $G$ is called a \emph{BCH Lie group} if $G$ is locally exponential and $\Lf (G)$ is BCH.
 \end{enumerate}
\end{defn}

We remark that contrary to claims in the BMS related physics literature (see \cite[p. 7]{Melas:2017jzb}) there are examples of infinite-dimensional Lie groups (beyond the Banach setting) which are locally exponential (as a trivial example the abelian Lie group $C^\infty (\SSS^2)$ has the identity as exponential function, whence it is even a global diffeomorphism). 
Note that every abelian Lie group and every Lie group modeled on a Banach space is BCH (\cite[Corollary IV.1.10]{Neeb06}.  
Due to the semidirect product structure, we can describe the exponential function of the BMS group, though it is unclear to us as to whether the BMS is even locally exponential:

\begin{setup}
 Following \cite{NaS13} (cf.~also \cite[II.5.9]{Neeb06}, the Lie group exponential of the BMS group is given by 
 \begin{align}
\exp \colon \mathcal{S} \rtimes \Lf (\mathcal{M}) = \Lf(\BMS) \rightarrow \BMS = \mathcal{S} \rtimes_\sigma \mathcal{M}, (F,f) \mapsto (\beta(f)F,\exp_{\mathcal{M}} (f)) \label{expo:form}\\
\text{ with } \beta (f) := \int_0^1 \sigma^\vee(\exp_{\mathcal{M}} (sf)) \mathrm{d}s. \notag
 \end{align}
 where $\exp_{\mathcal{M}}$ is the Lie group exponential of the M\"{o}bius group (which can be identified (up to the quotient) with the matrix exponential $\mathfrak{sl}(2,\bC) \rightarrow \SL (2,\bC)$. 
 We remark that $\beta$ is an operator valued integral (which gives meaning as an integral after applying it to elements $F \in C^\infty (\SSS^2)$). 
\end{setup}

\begin{rem}
 The formula \eqref{expo:form} shows that the Lie group exponential of the BMS group is a local diffeomorphism onto an open neighborhood of the identity if and only if the operator valued integrals $\beta(f)$ are (smoothly) invertible. This unfortunately turns out to be a complicated problem (cf.~e.g.~\cite[Problem IV.4]{Neeb06} on a related question). Note that it is neither clear that $\beta(f)$ is an invertible operator (say for $f \in \Lf(\mathcal{M})$ in some $0$-neighborhood) nor that the assignment $(f,F) \mapsto \beta(f)^{-1}F$ is smooth.  
 
As already reported, we were unfortunately not able to solve the question, whether the BMS group is locally exponential. However, based on the observations for the three-dimensional case and the Newman--Unti group, we propose:

\begin{conj} \label{conj:BMS_not_locally_exponential}
	The BMS group is not locally exponential.
\end{conj}

Evidence toward this claim is the observation that the BMS group is not analytic by \Cref{prop:BMS:notanalytic}. Thus \cite[Theorem IV.1.8]{Neeb06} implies that $\BMS$ can not be locally exponential \textbf{or} the Lie algebra $\Lf (\BMS)$ can not be BCH (i.e.\ since $\Lf(\BMS)$ is a \Frechet~space, this means that the BCH series can not converge on any neighborhood of $0$). Thus if one could show that the BCH-series converges in a $0$-neighborhood, this would imply that the BMS group can not be locally exponential.  
\end{rem}

In any case, $\BMS$ can not be a BCH Lie group and in particular there is no way, the BCH series can be used as a local model for the multiplication of the $\BMS$ group.
However, we have already seen in \Cref{prop:BMS_regular} that the BMS group is $C^0$-regular. This property allows us to recover the validity of certain familiar product formulae:

\begin{cor} \label{cor:BMS_Trotter}
 The $\BMS$ group satisfies the strong Trotter property, i.e.~for every $C^1$-curve $\gamma \colon [0,1] \rightarrow \BMS$ with $\gamma (0) = (0,\id)$ we have 
 $$(\gamma(t/n))^n \rightarrow \exp (t\gamma'(0)), \text{ as } n \rightarrow \infty.$$
 where convergence here is uniform in $t$ on compact subsets. In particular, the $\BMS$ satisfies the Trotter property, i.e.~$\forall v,w \in \Lf (\BMS)$ the following holds (again with uniform convergence in $t$ on compact subsets)
 $$\lim_{ n \rightarrow \infty} (\exp (tv/n) \exp(tw/n))^n = \exp (t(v+w)).$$
 \end{cor}

 \begin{proof}
  It suffices to establish the strong Trotter property as the Trotter property follows from it by setting $\gamma (t) = \exp(tv)\exp(tw)$. Now the strong Trotter property follows from $C^0$-regularity of the $\BMS$ group by \cite[Theorem 1]{Han20}.
 \end{proof}

Moreover, it is known that the strong Trotter property implies the (strong) commutator property \cite[Theorem H]{glo15}. We thus deduce again from $C^0$-regularity of the $\BMS$ group:

\begin{cor} \label{cor:BMS_strong_commutator}
 The $\BMS$ group has the strong commutator property, i.e. for every pair 
 of $C^1$-curve $\gamma , \eta \colon [0,1] \rightarrow \BMS$ with $\gamma(0)=(0,\id)=\eta(0)$ we have for $n \rightarrow \infty$ 
 $$(\gamma (\sqrt{t}/n)\eta(\sqrt{t}/n)(\gamma (\sqrt{t}/n))^{-1}(\eta(\sqrt{t}/n))^{-1})^{n^2} \rightarrow \exp (t\LB[\gamma'(0),\eta'(0)]).$$
 (uniformly in $t$ on compact subsets). Thus in particular for $v,w\in \Lf(\BMS)$ the familiar commutator formula holds
 $$(\exp(\sqrt{t}v/n)\exp(\sqrt{t}w/n)(\exp(\sqrt{t}v/n))^{-1}(\exp(\sqrt{t}w/n))^{-1})^{n^2} \rightarrow \exp(t\LB[v,w])$$
\end{cor}

Both the Trotter product formula and the commutator formula are useful in representation theory (see e.g.~\cite{NaS13}) and should thus be of interest in the finer analysis of representations of the $\BMS$ group.

\subsection{Extended and generalized BMS group} \label{ssec:eBMS_and_gBMS}

In this section we briefly discuss recently introduced variants of the BMS group together with their properties. They arise from the original BMS group by replacing the Lorentz group \(\Lor\) in the semidirect product structure of Equation~\eqref{eqn:BMS_semidirect_product} by an appropriate, larger, symmetry group, which is conveniently called `superrotations' \(\mathcal{R}\), cf.\ \cite{Barnich:2011ct}. Most prominently, we mention
\begin{itemize}
	\item the extended BMS (eBMS) due to Barnich and Troessaert \cite{Barnich:2009se,Barnich:2010eb,Barnich:2011mi}, and
	\item the generalized BMS (gBMS) due to Campiglia and Laddha \cite{Campiglia:2014yka,Campiglia:2015yka,Campiglia:2015kxa}.
\end{itemize}
The extended BMS uses two copies of the Bott--Virasoro group $\widehat{\Diff} (\SSS^1)$ as superrotations, \(\mathcal{R} := \widehat{\Diff} (\SSS^1) \times \widehat{\Diff} (\SSS^1)\) (here the notation indicates that the Bott--Virasoro group is a central extension of the circle diffeomorphisms $\Diff (\SSS^1)$, cf.~\cite[II.2]{KaW09}). For the generalized BMS, the superrotations are given by the diffeomorphism group of the two-sphere, \(\mathcal{R} := \operatorname{Diff} ( \mathbb{S}^2 )\). The physical difference lies in the precise definition of asymptotic flatness that is preserved by the respective asymptotic symmetry group. We remark that while the eBMS emerged from the context of the AdS-CFT correspondence, the gBMS was constructed with gravitational soft-scattering theorems in mind.
Finally we mention that in \cite{HaHaPaB17} an even more involved extension of the BMS group, the so called conformal BMS group (or cBMS) was proposed. Investigating the eBMS and cBMS is beyond the scope of the present article. However, as the  generalized BMS group is relatively easy to describe, we will now describe the associated Lie theory along the lines of our treatment of the BMS group.

\subsubsection*{Lie theory for the generalized BMS group}

\begin{defn}
 The generalized BMS group is the semidirect product 
 $$\gBMS := \mathcal{S} \rtimes_\alpha \Diff (\SSS^2)$$
 where $\Diff (\SSS^2)$ denotes the (Lie group) of smooth diffeomorphisms of the sphere, in the literature known as `superrotations', and the action inducing the semidirect product is 
 $$\alpha \colon \mathcal{S} \times \Diff (\SSS^2 ) \rightarrow \mathcal{S},\quad (F,\phi ) \mapsto F\circ \phi.$$
\end{defn}
Since general diffeomorphisms of the sphere do not preserve the conformal structure this structure is not preserved by elements in the $\gBMS$. However, it is easy to see that the $\gBMS$ is an infinite-dimensional Lie group

\begin{prop}\label{gBMS}
 The group $\gBMS$ is an infinite-dimensional \Frechet~Lie group which is $C^0$-regular. Moreover, its Lie algebra is the semidirect product $\mathcal{S} \rtimes_{\hat{\alpha}} \mathfrak{X}(\SSS^2)$, where $\hat{\alpha}$ is the derived action of the Lie algebra of vector fields on the sphere on the supertranslations.
\end{prop}

\begin{proof}
 Note first that the group $\Diff (\SSS^2)$ is an open subset of $C^\infty (\SSS^2,\SSS^2)$ and this structure turns it into an infinite-dimensional Lie group, \cite[Theorem 11.11]{Michor80}.
 The statement on the Lie group structure and the Lie algebra will follow at once, if we can prove that $\alpha$ is smooth. However, since $$\alpha \colon \mathcal{S} \times \Diff (\SSS^2) = C^\infty (\SSS^2) \times \Diff (\SSS^2) \rightarrow C^\infty (\SSS^2), (F,\phi) \mapsto F\circ \phi$$
 is simply the composition, smoothness follows at once from smoothness of the composition operation on spaces of smooth function spaces on compact manifolds (see e.g.~\cite[Theorem 11.4]{Michor80}).
 
 Concerning the $C^0$-regularity of $\gBMS$ we can argue as in the proof of \Cref{prop:BMS_regular}: Since $C^0$-regularity is an extension property, it suffices to prove that $\Diff (\SSS^2)$ is $C^0$-regular. However, this is well known (cf.~e.g.~\cite{Milnor,Sch15}).
\end{proof}

\begin{rem}
 Since $\Diff (\SSS^2)$ is not locally exponential, also $\gBMS$ can not be a locally exponential Lie group. However, as seen in the section on the BMS group, we can deduce from the $C^0$-regularity of the group $\gBMS$ the following information: $\gBMS$ satisfies the (strong) Trotter condition and the (strong) commutator formula holds on $\gBMS$. 
\end{rem}

Due to the nature of the group product of the generalized BMS group there seems to be no natural way to embed the BMS group as a subgroup. It is of course straight forward to check that the generalized BMS group contains the BMS group as a set. However this inclusion does not respect the group product of the BMS. Since general diffeomorphisms of a manifold do not respect conformal structures, there seems to be no way to remedy this defect of the generalized BMS group by changing for example the action inducing the semidirect product.

Note that the situation is different for the so called extended BMS group proposed in \cite{Barnich:2009se}, see also \cite{Ruzz20}. It can be shown that the extended BMS group arises as a group preserving the conformal structure of the punctured extended complex plane. Thus the group product should be compatible with the natural embedding of the BMS group. However a detailed study of the extended BMS group as a Lie group is beyond the scope of the present article.

\section{Conclusion}

We have studied the BMS group from the viewpoint of infinite-dimensional Lie group theory. In particular, we have shown that the BMS group is not analytic (\propref{prop:BMS:notanalytic}). This implies that the BMS group is either not locally exponential or its Lie algebra is not BCH, which leads us to conjecture the first (Conjecture~\ref{conj:BMS_not_locally_exponential}). Furthermore, we have established that the BMS group is \(C^0\)-regular in the sense of Milnor (\propref{prop:BMS_regular}) and that it satisfies the strong Trotter property (\colref{cor:BMS_Trotter}) as well as the strong commutator property (\colref{cor:BMS_strong_commutator}). Furthermore, we discuss the cases of the extended BMS group and the generalized BMS group (Subsection~\ref{ssec:eBMS_and_gBMS}) and provide an outlook to the much more involved case of the Newman--Unti group (Subsection~\ref{ssec:outlook_NU}), which will be studied in \cite{Prinz_Schmeding_2}.

\begin{appendix}
\section{A mathematics-physics dictionary}
As the goal of the present article is to provide a bridge between objects of interests to the physicist and techniques from (infinite-dimensional) differential geometry, we provide a dictionary for the readers convenience:

\begin{table}[ht]
\caption{\textbf{Mathematics-physics dictionary}}
\begin{tabularx}{\textwidth}{|p{4.6cm}|X|}\toprule
 \textbf{Object} & \textbf{Definition / Explanation}  \\  \midrule
 asymptotically flat spacetime & spacetime, where the metric tensor \(g_{\mu \nu}\) decays to the Minkowski metric tensor \(\eta_{\mu \nu}\) at `infinity' \\ \cmidrule(r){1-1}\cmidrule(lr){2-2}
  Baker--Campbell--Hausdorff (BCH) formula & $\exp(x)\exp(y) = \exp(\text{BCH}(x,y))$ for Lie algebra elements $x,y$, with $\text{BCH}(x,y) = \sum_{n=1}^\infty H_n(x,y)$ where $H_1 (x,y) = x+y, H_2 = \frac{1}{2} \LB[x,y] , \ldots$  \\ \cmidrule(r){1-1}\cmidrule(lr){2-2}
  Bondi--Metzner--Sachs (BMS) group & group of gauge fixing diffeomorphisms, semidirect product of the supertranslations and the (proper orthochronous) Lorentz group $\operatorname{BMS} := \mathcal{S} \rtimes \Lor$\\ \cmidrule(r){1-1}\cmidrule(lr){2-2}
  $\Evol \; ( \equiv \mathcal{T} \exp ) \colon $ $C([0,1],\Lf(G)) \rightarrow C^1([0,1],G)$ & evolution operator (aka ``time ordered exponential''), sends a Lie algebra valued curve to the solution of the associated Lie type differential equation \\ \cmidrule(r){1-1}\cmidrule(lr){2-2}
  $\evol \colon C([0,1],\Lf(G)) \rightarrow G$ & $\evol (\gamma) = \Evol(\gamma)(1)$ \\\cmidrule(r){1-1}\cmidrule(lr){2-2}
  generalized BMS group & $\gBMS := \mathcal{S} \rtimes_\alpha \Diff (\SSS^2)$, \Cref{gBMS}\\ \cmidrule(r){1-1}\cmidrule(lr){2-2}
  analytic Lie group & Lie group which is an analytic manifold (see Definition \ref{defn:real_analytic}) with analytic group operations.\\ \cmidrule(r){1-1}\cmidrule(lr){2-2}
  infinite-dimensional Lie group & a group which is a smooth manifold modeled on a locally convex space such that the group operations are smooth maps as in \Cref{defn: deriv}\\ \cmidrule(r){1-1}\cmidrule(lr){2-2}
 locally exponential Lie group & Lie group which admits exponential coordinates (i.e. the Lie group exponential is a local diffeomorphism near the unit)\\ \cmidrule(r){1-1}\cmidrule(lr){2-2}
 $(C^0-)$regular Lie group & Lie group for which $\evol$ exists and is smooth\\ \cmidrule(r){1-1}\cmidrule(lr){2-2}
 Newman-Unti (NU) group & $\operatorname{NU} := \mathcal{N} \rtimes \Lor$, \cite{Prinz_Schmeding_2}\\ \cmidrule(r){1-1}\cmidrule(lr){2-2}
 supertranslations & $\mathcal{S} := C^\infty (\SSS^2) \equiv C^\infty (\SSS^2,\R)$\\  \cmidrule(r){1-1}\cmidrule(lr){2-2}
 Trotter product formula & $\lim_{ n \rightarrow \infty} (\exp (tv/n) \exp(tw/n))^n = \exp (t(v+w))$\\ 
 \bottomrule
 \end{tabularx}
\end{table}

\clearpage

\section{Essentials from infinite-dimensional calculus}\label{app:idim}

In this appendix we collect background on infinite-dimensional calculus and Lie groups. We follow
\cite{Neeb06}. In the following we will always pick a field $\K \in \{ \R, \bC\}$ to be either the real or complex numbers. 
Then $\K^\times$ denotes the invertible elements in the field $\K$.
First of all we define a notion of (continuously) differentiable map which works in arbitrary locally convex spaces and is not restricted to manifolds modeled on Banach spaces. 

\begin{defn}\label{defn: deriv} Let $E, F$ be locally convex spaces over a field $\K$, $U \subseteq E$ be an open subset,
$f \colon U \rightarrow F$ a map and $r \in \N_{0} \cup \{\infty\}$. If it
exists, we define for $(x,h) \in U \times E$ the directional derivative
$$df(x,h) := D_h f(x) := \lim_{\K^\times  \ni t\rightarrow 0} t^{-1} \big(f(x+th) -f(x)\big).$$
We say that $f$ is $C^r$ if the iterated directional derivatives
\begin{displaymath}
d^{(k)}f (x,y_1,\ldots , y_k) := (D_{y_k} D_{y_{k-1}} \cdots D_{y_1}f) (x)
\end{displaymath}
exist for all $k \in \N_0$ such that $k \leq r$, $x \in U$ and
$y_1,\ldots , y_k \in E$ and define continuous maps
$d^{(k)} f \colon U \times E^k \rightarrow F$. If $f$ is $C^\infty_\R$, it is called \emph{smooth}.
 If $f$ is $C^\infty_\bC$, it is said to be \emph{complex analytic} or \emph{holomorphic} and that $f$ is of class $C^\omega_\bC$.\footnote{Recall from \cite[Proposition 1.1.16]{dahmen2011}
 that $C^\infty_\bC$ functions are locally given by series of continuous homogeneous polynomials (cf.\ \cite{BS71a,BS71b}).
 This justifies the abuse of notation.} 
\end{defn}

\begin{defn}[Complexification of a locally convex space]
 Let $E$ be a real locally convex topological vector space. Endow the locally convex product $E_\bC := E \times E$ with the following operation
 \begin{displaymath}
  (x+iy).(u,v) := (xu-yv, xv+yu), \quad \forall x,y \in \R,\; u,v \in E.
 \end{displaymath}
 The complex vector space $E_\bC$ is called the \emph{complexification} of $E$. Identify $E$ with the closed real subspace $E\times \{0\}$ of $E_\bC$.
\end{defn}

\begin{defn}	\label{defn:real_analytic}							
 Let $E$, $F$ be real locally convex spaces and $f \colon U \rightarrow F$ defined on an open subset $U \subseteq E$.
 $f$ is called \emph{real analytic} (or $C^\omega_\R$) if $f$ extends to a $C^\infty_\bC$-map $\tilde{f}\colon \tilde{U} \rightarrow F_\bC$ on an open neighborhood $\tilde{U}$ of $U$ in the complexification $E_\bC$.
\end{defn}

For $r \in \mathbb{N}_0 \cup \{\infty, \omega\}$, being of class $C^r_\K$ is a local condition,
i.e.\ if $f|_{U_\alpha}$ is $C^r_\K$ for every member of an open cover $(U_\alpha)_{\alpha}$ of its domain,
then $f$  is $C^r_\K$. (See \cite[pp. 51-52]{Glock02} for the case of $C^\omega_\R$, the other cases are clear by definition.)
In addition, the composition of $C^r_\K$-maps (if possible) is again a $C^r_\K$-map (cf. \cite[Propositions 2.7 and 2.9]{Glock02}).

\begin{setup}[$C^r$-Manifolds and $C^r$-mappings between them]
For $r \in \N_0 \cup \{\infty, \omega\}$, manifolds modeled on a fixed locally convex space can be defined as usual.
The model space of a locally convex manifold and the manifold as a topological space will always be assumed to be Hausdorff spaces.
However, we will not assume that infinite-dimensional manifolds are second countable or paracompact.
We say $M$ is a \emph{Banach} (or \emph{Fr\'{e}chet}) manifold if all its modeling spaces are Banach (or
Fr\'{e}chet) spaces.
\end{setup}
\begin{setup}
Direct products of locally convex manifolds, tangent spaces and tangent bundles as well as $C^r$-maps between manifolds may be defined as in the finite-dimensional setting.
For $C^r$-manifolds $M,N$ we use the notation $C^r(M,N)$ for the set of all $C^r$-maps from $M$ to $N$.
Furthermore, we define \emph{locally convex Lie groups} as groups with a $C^\infty$-manifold structure turning the group operations into $C^\infty$-maps.
\end{setup}

\begin{setup}\label{defn:CRS}
Let $E_1$, $E_2$ and $F$ be locally convex spaces and $f\colon U\times V\to F$
be a mapping on a product of locally convex subsets $U\subseteq E_1$ and $V\subseteq E_2$
be open. Given $k,\ell\in\N_0\cup\{\infty\}$, we say that $f$ is $C^{k,\ell}$
if $f$ is continuous and there exist continuous mappings
$d^{\,(i,j)}f\colon U\times V\times E_1^i\times E_2^j\to F$
for all $i,j\in\N_0$ such that $i\leq k$ and $j\leq \ell$
such that
\[
d^{\,(i,j)}f(x,y,v_1,\ldots,v_i,w_1,\ldots,w_j)=(D_{(v_i,0)}\cdots D_{(v_1,0)}
D_{(0,w_j)}\cdots D_{(0,w_1)}f)(x,y)
\]
for all $x\in U^0$, $y\in V^0$ and $v_1,\ldots, v_i\in E_1$,
$w_1,\ldots, w_j\in E_2$.
\end{setup}
One can also define $C^{k,\ell}$-maps $M_1\times M_2\to N$ if~$M_1$ is a $C^k$-manifold, $M_2$ a $C^\ell$-manifold and~$N$ a $C^{k+\ell}$-manifold,
checking the property in local charts. Moreover, one can also consider mappings which are complex differentiable in one component and only real differentiable in the other. 
We refer to \cite{AS15} for more information. 

\begin{rem}
 The calculus used in the present article is also known as ``Bastiani-calculus''. It is not the only valid choice of calculus beyond the realm of Banach spaces. An alternative approach is the so called ``convenient calculus'' (see \cite{KaM97}). In general both calculi do not coincide outside of \Frechet~spaces (where the Bastiani calculus is stronger as it requires continuity of the underlying mappings). However, in the present paper we are exclusively working with \Frechet~spaces and manifolds modeled on \Frechet~spaces. Thus there is no difference between our calculus and the convenient calculus in this setting.  
\end{rem}

\end{appendix}

\addcontentsline{toc}{section}{References}
\bibliography{BMS_project_v2}
\end{document}